\documentclass[12pt]{amsart}
\usepackage[cp1251]{inputenc}
\usepackage[english]{babel}
\usepackage{amsaddr}
\usepackage{amsmath}
\usepackage{amssymb}
\usepackage{amsfonts}
\usepackage{epsfig}
\usepackage{srcltx}
\usepackage{subfigure}
\usepackage{cite}
\usepackage{float}
\usepackage[a4paper,  mag=1000, includefoot,  left=2cm, right=2cm, top=2cm, bottom=2cm, headsep=1cm, footskip=1cm ]{geometry}
\usepackage[unicode,colorlinks]{hyperref}
\hypersetup{colorlinks=true, citecolor=blue, linkcolor=blue}

\newtheorem{Th}{Theorem}

\begin{document}
\thispagestyle{empty}

\title[Asymptotic analysis of the autoresonance capture in oscillating systems]
{Asymptotic analysis of the autoresonance capture in oscillating systems with combined excitation}

\author{Oskar Sultanov}

\address{Institute of Mathematics, Ufa Federal Research Center, Russian Academy of Sciences, \newline 112, Chernyshevsky str., Ufa, Russia, 450008}
\email{oasultanov@gmail.com}


\maketitle {\small
\begin{quote}
\noindent{\bf Abstract.}
A mathematical model of autoresonance in nonlinear systems with combined parametric and external chirped frequency excitation is considered. Solutions with a growing amplitude and a bounded phase mismatch are associated with the autoresonant capture. By applying Lyapunov function method we investigate the conditions for the existence and stability of autoresonant modes and construct long-term asymptotics for stable solutions. In particular, we show that unstable regimes become stable when the system parameters pass through certain threshold values.
\medskip


\end{quote}
}

\section*{Introduction}
Autoresonance is a phenomenon of persistent phase synchronization between an oscillatory nonlinear system and a small resonant chirped frequency perturbation that leads to a significant increase in the amplitude of the oscillator~\cite{LF09}. Autoresonance was fist realized in relativistic particles accelerators in the middle of the twentieth century, and nowadays, it is considered as a universal phenomenon that occurs in a wide range of physical systems~\cite{AN87,FF01,USM10,KHM14,BShF18}. In recent times, the mathematical models of autoresonance are actively studied numerically and analytically (see~\cite{LK08}, and references therein). However, up to now the autoresonance in nonlinear systems with parametric~\cite{KhM01,KG07} and external~\cite{LF08} driving has been treated separately. To the best of our knowledge, the effect of the combined excitation on the capture into autoresonance has not previously been discussed. In this paper, parametric and external perturbations are considered together, and all possible autoresonant modes are described with a help of the Lyapunov function technique.

The paper is organized as follows. In section \ref{sec1}, the mathematical formulation of the problem is given. In section \ref{sec2} particular isolated autoresonant solutions are constructed. Section \ref{sec3} provides linear and nonlinear stability analysis of these solutions. Asymptotic analysis of general autoresonant solutions is contained in section \ref{sec4}. The paper concludes with a brief discussion of the results obtained.

\section{Problem statement}
\label{sec1}
Consider the system of two differential equations:
\begin{gather}
    \label{MS}
    \begin{split}
        \frac{d\rho}{d\tau}+\mu(\tau) \rho \sin (2\psi+\nu) =  \sin\psi, \\
        \rho\Big[\frac{d\psi}{d\tau}-\rho^2+ \lambda \tau+ \mu(\tau) \cos (2\psi+\nu)\Big]= \cos\psi,
    \end{split}
\end{gather}
with the parameters $\lambda\neq 0$, $\nu\in [0,\pi)$ and a smooth bounded given function $\mu(\tau)$. This model system arises in the study of the autoresonance phenomena in a wide class of nonlinear oscillators with a time-periodic slowly-varying perturbation and describes the principal terms of the asymptotic behaviour of the oscillators at the initial step of the capture. We assume that the function $\mu(\tau)$, corresponding to the amplitude of a parametric driving, has a power-law asymptotics:
\begin{gather*}
    \mu(\tau)=\mu_0\tau^{-1/2}+\sum_{k=1}^\infty \mu_k \tau^{-(2k+1)/2},\quad \tau\to\infty, \quad \mu_k={\hbox{\rm const}}.
\end{gather*}
The unknown functions $\rho(\tau)$ and $\psi(\tau)$ play the role of the amplitude and the phase mismatch, respectively. The solutions with $\rho(\tau)\to \infty$ and $\psi(\tau)=\mathcal O(1)$ as $\tau\to\infty$ are associated with the phase-locking~\cite{PRK02} and the autoresonance~\cite{LF09}. In addition to autoresonant modes, system \eqref{MS} also has solutions with a limited amplitude and an increasing phase mismatch. Such solutions correspond to the phase-slipping phenomenon and are not investigated in the framework of the present paper.

Note that \eqref{MS} with $\mu\equiv 0$ corresponds to the case of a pure external excitation, see \cite{LK03,LF08}. If $\mu\equiv 1$, the action of the external driving becomes insignificant and system \eqref{MS} reduces to a weak perturbation of the model of parametric autoresonance~\cite{KhM01,OS16}. The combined excitation takes effect only in the case when $\mu\sim\tau^{-1/2}$.

In this paper, we investigate the conditions for the existence and stability of autoresonant solutions to system \eqref{MS} and construct long-term asymptotics for such solutions. In the first step, particular autoresonant solutions with special power-law asymptotics at infinity are constructed and their Lyapunov stability is investigated. Linear stability analysis allows us to identify only the conditions that guarantee the instability of solutions; the stability at this level cannot be determined due to certain features of the equations. To solve the stability problem, it is necessary to carry out a careful nonlinear analysis with the Lyapunov function method. The presence of stability will ensure the existence of a two-parameter family of autoresonant solutions. For such solutions, the asymptotics at infinity will be constructed at the last step.

As but one example we consider Duffing's oscillator with small combined parametric and external excitation:
\begin{gather}
    \label{ex}
        \frac{d^2u}{dt^2}+(1+\varepsilon B(t)) (u-\gamma \varepsilon u^3)=\varepsilon A(t),
\end{gather}
where $A(t)=\cos\phi(t)$, $B(t)=\beta(1+\varepsilon t)^{-1/2}\cos(2\phi(t)+\nu)$, $\phi(t)=t-\alpha t^2$, $0<\varepsilon,\alpha\ll 1$, $\beta,\gamma={\hbox{\rm const}}$, $\gamma>0$. It is easy to see that equation \eqref{ex} without perturbation ($A(t)\equiv B(t)\equiv 0$) has stable fixed point $(u,u')=(0,0)$ and periodic solutions. Consider the solutions of the perturbed equation with small enough initial values $(u(0),u'(0))$. Numerical analysis shows that for some initial data there are solutions whose the energy $E(t)\equiv (u(t))^2/2-\gamma \varepsilon (u(t))^4/4+(u'(t))^2/2$ significantly increases with time and the phase $\Phi(t)\equiv \arctan \big(u'(t)/u(t)\big)$ is synchronised with the pumping:  $\Delta(t)\equiv\phi(t)+\Phi(t)=\mathcal O(1)$. Such solutions are associated with the capture into the autoresonance. For non-captured solutions, the energy remains small and the phase mismatch increases:  $|\Delta|\to\infty$, see. Fig.~\ref{Pic1}. For the description of a long-term evolution of solutions to equation \eqref{ex}, the method of two scales can be used. We introduce a slow time $\tau=\varepsilon t/(2\kappa)$, ($\kappa=(4/3\gamma)^{1/3}$) and a fast variable $\phi=\phi(t)$, then the asymptotic substitution $u(t)=\kappa \rho(\tau) \cos \big(-\phi+\psi(\tau)\big) +\mathcal O(\varepsilon)$ in equation \eqref{ex} and the averaging over the fast variable $\phi$ in the leading-order term in $\varepsilon$ lead to system \eqref{MS} for the slowly varying functions $\rho(\tau)$ and $\psi(\tau)$ with  $\lambda=8\alpha \varepsilon^{-2} \kappa^2$ and $\mu(\tau)=\beta(1+2\kappa\tau)^{-1/2}\sqrt{2\kappa}/4$. Similarly, system \eqref{MS} appears in many other problems of nonlinear physics related to autoresonance.

\begin{figure}
\centering
\includegraphics[width=0.45\linewidth]{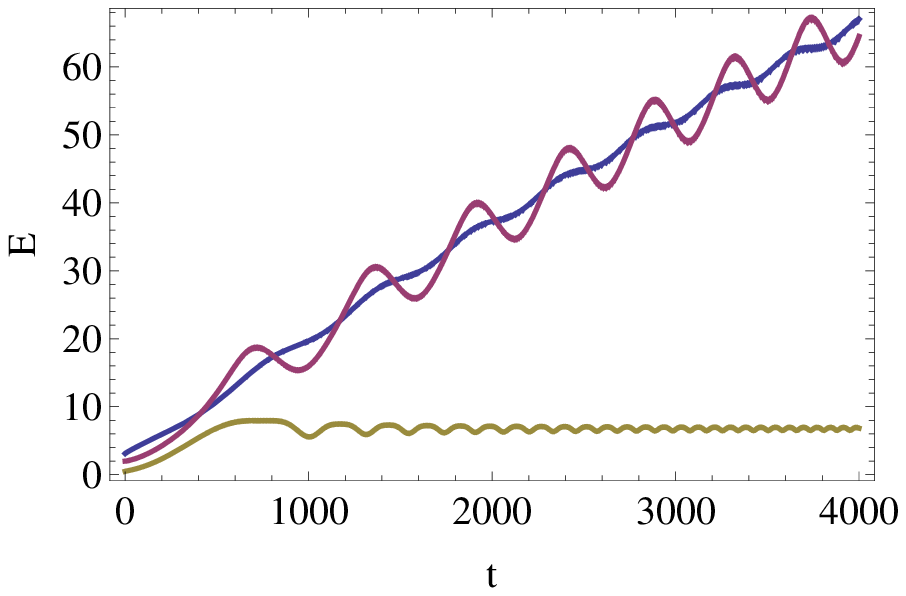}
\hspace{4ex}
\includegraphics[width=0.45\linewidth]{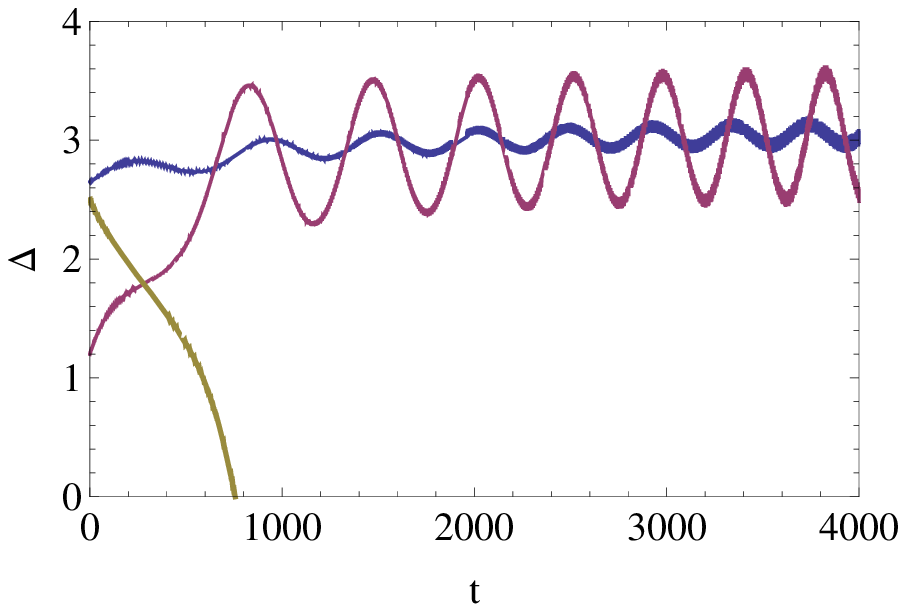}
\caption{The evolution of $E(t)$ and $\Delta(t)$ for solutions of \eqref{ex} with $\varepsilon=10^{-2}$, $\alpha=0.25\cdot 10^{-4}$, $\beta=1$, $\gamma=1/6$, $\nu=0$ and various initial data.} \label{Pic1}
\end{figure}

\section{Particular autoresonant solutions}
\label{sec2}
The simplest asymptotic expansion for autoresonant solutions is constructed in the form of power series with constant coefficients:
\begin{gather}
    \label{asser}
        \rho_\ast(\tau)=\rho_{-1}\tau^{1/2}+\rho_0+\sum_{k=1}^{\infty} \rho_k \tau^{-k/2}, \quad \psi_\ast(\tau)=\psi_0+\sum_{k=1}^\infty\psi_k\tau^{-k/2}, \quad \tau\to\infty.
\end{gather}
Substituting this series in system \eqref{MS} and grouping the terms of the same power of
$\tau$, we can determine all the coefficients $\rho_k$ and $\psi_k$. In particular, $\rho_{-1}= \lambda^{1/2}$, $\rho_0=0$, and $\psi_0$ satisfies the following equation:
\begin{gather}
    \label{teq}
    \mathcal P(\psi_0;\delta,\nu)\equiv \delta\sin (2\psi_0+\nu)-\sin\psi_0=0,  \quad \delta=\mu \lambda^{1/2},
\end{gather}
that has a different number of roots depending on the values of the parameters $\delta$ and $\nu$.
Besides, if the inequality $\mathcal P'(\psi_0;\delta,\nu)\neq 0$ holds, the remaining coefficients $\rho_k$, $\psi_k$ as $k\geq 1$ are determined from the chain of linear equations:
\begin{gather}
    \label{sysFG}
        \begin{split}
            2 \rho_{-1} \rho_k & =\mathcal F_k(\rho_0,\dots,\rho_{k-1}, \psi_0,\dots,\psi_{k-2}),  \\
            \mathcal P'(\psi_0;\delta,\nu)\psi_k & = \mathcal G_k(\rho_{-1},\dots,\rho_{k-1}, \psi_0,\dots,\psi_{k-1}),
    \end{split}
\end{gather}
where
\begin{align*}
    &\mathcal F_1=0, \ \ \mathcal F_2=(\delta \cos (2\psi_0+\nu)-\cos\psi_0)\lambda^{-1/2},\ \ \mathcal F_3=-\psi_1 (2 \delta \sin(2\psi_0+\nu)-\sin\psi_0)\lambda^{-1/2}, \\
    & \mathcal G_1=-\frac{\rho_{-1}}{2},\  \ \mathcal G_2=-\mathcal P''(\psi_0;\delta,\nu)\frac{\psi_1^2}{2}-\mu_2 \rho_{-1}\sin(2\psi_0+\nu), \\
    & \mathcal G_3=  -\mathcal P''(\psi_0;\delta,\nu)\psi_1\psi_2 - \mathcal P'''(\psi_0;\delta,\nu) \frac{\psi_1^3}{6}-\mu \rho_2\sin(2\psi_0+\nu),
\end{align*}
etc. Note that the pair of equations $\mathcal P(\psi_0;\delta,\nu)=0$ and $\mathcal P'(\psi_0;\delta,\nu)=0$ defines the bifurcation curves
\begin{gather*}
    \Gamma_{-}\stackrel{def}{=}\{(\delta,\nu)\in\mathbb R\times[0,\pi): \ell(\delta,\nu)=0, \delta<0\}, \quad \Gamma_{+}\stackrel{def}{=}\{(\delta,\nu)\in\mathbb R\times[0,\pi): \ell(\delta,\nu)=0, \delta>0\}
\end{gather*}
on the parameter plane $(\delta,\nu)$, where $\ell(\delta,\nu) \equiv (4\delta^2-1)^3-27 \delta^2\sin^2\nu$. It can easily be checked that system \eqref{sysFG} is solvable whenever $(\delta,\nu)\not\in\Gamma_{\pm}$. The bifurcation curves divide the parameter plane into the following parts (see Fig.~\ref{figom}):
\begin{align*}
   & \Omega_{-}\stackrel{def}{=}\{(\delta,\nu)\in\mathbb R\times[0,\pi): \ell(\delta,\nu)<0\}, \quad
   & \Omega_{+}\stackrel{def}{=}\{(\delta,\nu)\in\mathbb R\times[0,\pi): \ell(\delta,\nu)>0\}.
\end{align*}
If $(\delta,\nu)\in\Omega_{-} $, the equation $\mathcal P(\psi_0;\delta,\nu)=0$ has four different roots on the interval $[0,2\pi)$. In the case $(\delta,\nu)\in\Omega_{+}$, there are only two different roots.
\begin{figure}
\centering
\includegraphics[width=0.45\linewidth]{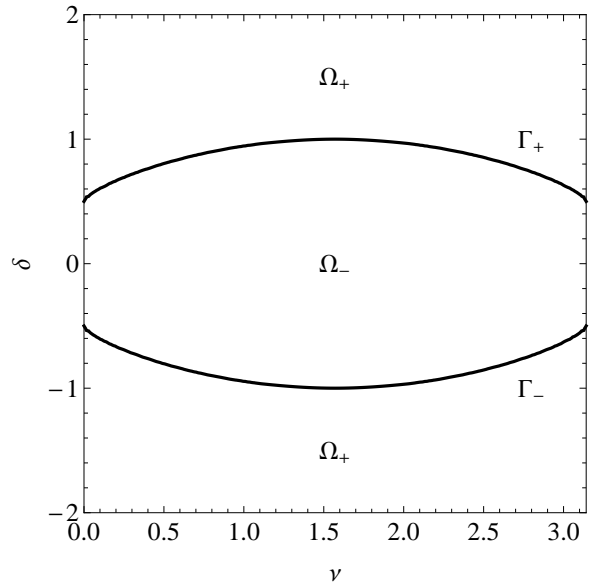}
\caption{Partition of the parameter plane.} \label{figom}
\end{figure}
Thus we have
\begin{Th}
    If $(\delta,\nu)\in\Omega_{+}$, system \eqref{MS} has 4 different solutions with asymptotic expansion in the form of a series \eqref{asser}.
    If $(\delta,\nu)\in\Omega_{-}$, system \eqref{MS} has 2 different solutions with asymptotic expansion in
the form of a series \eqref{asser}.
\end{Th}
The existence of solutions $\rho_\ast(\tau)$, $\psi_\ast(\tau)$ with the asymptotics \eqref{asser} as $\tau\geq \tau_\ast$ follows from~\cite{AN89,KF13}, and comparison theorems~\cite{LK14} applied to system \eqref{MS} ensure that the solutions can be extended to the semi-axis. In the next sections, the stability of such solutions is discussed, as well as the asymptotics for two-parameter family of autoresonant solutions to system \eqref{MS} is constructed.

As an example, let us fix $\nu=0$, then $\Gamma_\pm=(\pm 1/2, 0)$, $\Omega_-=\{|\delta|<1/2,\nu=0\}$, $\Omega_+=\{|\delta|>1/2,\nu=0\}$. For $(\delta,\nu)\in\Omega_+$, equation \eqref{teq} has 4 roots: $\{0,\pi,\pm\arccos(2\delta)^{-1}\}$; in the opposite case $(\delta,\nu)\in\Omega_+$, there are only 2 roots: $\{0,\pi\}$.

\section{Stability analysis}
\label{sec3}

1. {\bf Linearization}. Let $\rho_\ast(\tau)$, $\psi_\ast(\tau)$  be one of the particular autoresonant solutions with asymptotics \eqref{asser}.
Then, the change of variables $\rho(\tau)=\rho_\ast(\tau)+r(\tau)$, $\psi(\tau)=\psi_\ast(\tau)+p(\tau)$ leads to the following system
\begin{align}
\label{eq00}
\begin{split}
&\frac{dr}{d\tau}=\sin(\psi_\ast+p)-\sin\psi_\ast + m(\tau)\big[\rho_\ast\sin(\psi_\ast+\nu)- (\rho_\ast+r)\sin(2\psi_\ast+2p+\nu)\big], \\
&\frac{dp}{d\tau}=2\rho_\ast r+\frac{\cos(\psi_\ast+p)}{\rho_\ast+r}-\frac{\cos\psi_\ast}{\rho_\ast}+ m(\tau)\big[\cos(2\psi_\ast+\nu)-\cos(\psi_\ast+p+\nu)\big],
\end{split}
\end{align}
with a fixed point at $(0,0)$. Consider the linearized system:
\begin{gather*}
    \frac{d}{d\tau}\begin{pmatrix} r \\ p \end{pmatrix} = {\bf A}(\tau) \begin{pmatrix} r \\ p \end{pmatrix}, \ \
    {\bf A}(\tau)=
        \begin{pmatrix}
            \displaystyle  -m(\tau) \sin(2\psi_\ast+\nu) & \displaystyle \cos\psi_\ast-2  \rho_\ast m(\tau) \cos(2\psi_\ast+\nu) \\
            \displaystyle 2\rho_\ast - \frac{\cos\psi_\ast}{\rho_\ast^2} & \displaystyle 2 m(\tau)\sin(2\psi_\ast+\nu)-\frac{\sin\psi_\ast}{\rho_\ast} \end{pmatrix}.
\end{gather*}
It can easily be checked that the roots $z_{\pm}(\tau)$ of the characteristic equation $|{\bf A}(\tau)-z {\bf I}|=0$ can be represented in the form $z_{\pm}(\tau)=x(\tau)\pm \sqrt{y(\tau)}$, where
\begin{gather*}
    \begin{split}
    & x(\tau)=-\frac{1}{4\tau}+\mathcal O(\tau^{-3/2}), \\
    & y(\tau)=(4\lambda\tau)^{1/2}\Big(-\mathcal P'(\psi_0;\delta,\nu)-\mathcal P''(\psi_0;\delta,\nu)\psi_1 \tau^{-1/2}+\mathcal O(\tau^{-1})\Big), \quad \tau\to\infty.
    \end{split}
\end{gather*}
Therefore, if $\mathcal P'(\psi_0;\delta,\nu)<0$, the leading asymptotic terms of the eigenvalues $z_\pm(\tau)$ are real of different signs. This implies that the fixed point $(0,0)$ of \eqref{eq00} is a saddle in the asymptotic limit, and the corresponding autoresonant solutions $\rho_\ast(\tau)$, $\psi_\ast(\tau)$ to system \eqref{MS} are unstable.
In the opposite case, when $\mathcal P'(\psi_0;\delta,\nu)>0$, we have $z_\pm(\tau)=\pm i (4\lambda\tau)^{1/4} \sqrt{\mathcal P'(\psi_0;\delta,\nu)}+\mathcal O(\tau^{-1/4})$, $\Re z_\pm(\tau)=\mathcal O(\tau^{-1})$ as $\tau\to\infty$ and the fixed point $(0,0)$ is a center in the asymptotic limit. Note that in the case $\Re z_\pm(\tau)\to 0$ as $\tau\to\infty$, linear stability analysis fails and it is necessary to consider nonlinear terms of equations. Indeed, let us consider as an example the non-autonomous system of differential equations:
\begin{gather}
    \label{es}
    \frac{da}{dt}=a b t^{-1/4}-at^{-1},\quad \frac{db}{dt}=-\frac{b}{2}t^{-1}.
\end{gather}
The solution has the form: $a(t)=a_0 t^{-1}\exp(4 b_0 t^{1/4})$, $b(t)=b_0 t^{-1/2}$ with arbitrary parameters $a_0$ and $b_0$. These formulas indicate that the fixed point $(0,0)$ of system \eqref{es} is unstable. However, both roots of the characteristic equation for the corresponding linearized system are negative: $z_{\pm}(t)=(-3\pm 1)/(4t)<0$.

Thus we have
\begin{Th}
    If $\mathcal P'(\psi_0;\delta,\nu)<0$, the solution $\rho_\ast(\tau)$, $\psi_\ast(\tau)$ with asymptotics \eqref{asser} is unstable.
\end{Th}

2. {\bf Lyapunov function}. In the case when linear stability analysis fails, the stability problem can be solved using the Lyapunov function method. We have
\begin{Th}
If $\mathcal P'(\psi_0;\delta,\nu)>0$, then the solution $\rho_\ast(\tau)$, $\psi_\ast(\tau)$ with asymptotics \eqref{asser} is asymptotically stable. Moreover, there exist $d_0>0$ and $\tau_0>0$ such that for all $(r_0,\phi_0)$: $(r_0-\rho_\ast(\tau_0))^2+(\phi_0-\psi_\ast(\tau_0))^2<d_0^2$ the solution $\rho(\tau)$, $\psi(\tau)$ to system \eqref{MS} with initial data $\rho(\tau_0)=r_0$, $\psi(\tau_0)=\phi_0$ has the asymptotics:
\begin{gather}
\label{asth}
\rho(\tau)=\sqrt{\lambda\tau}+o(\tau^{-1/4}), \quad \psi(\tau)=\psi_0+o(1), \quad \tau>\tau_0.
\end{gather}
\end{Th}
 \begin{proof}
In system \eqref{MS} we make the change of variables
\begin{gather}
\label{exch2}
\rho(\tau)=\rho_\ast(\tau)+\omega_0 \tau^{-1/4} R(\eta),\quad \psi(\tau)=\psi_\ast(\tau)+\Psi(\eta), \quad \eta=\varkappa\tau^{5/4},
\end{gather}
$\omega_0= \big(\mathcal P'(\psi_0;\delta,\nu)\big)^{1/2}(4\lambda)^{-1/4}$, $\varkappa=4/5$ and for new functions $R(\eta)$, $\Psi(\eta)$ we study the stability of the fixed point $(0,0)$ for the following system
\begin{gather}
    \label{ham}
    \frac{dR}{d\eta}=-\partial_\Psi H(R,\Psi,\eta), \quad
    \frac{d\Psi}{d\eta}=\partial_R H(R,\Psi,\eta)+F(R,\Psi,\eta),
\end{gather}
where
\begin{eqnarray*}
    H  & = & \omega_0 \varkappa^{2/5} \eta^{-2/5} \rho_\ast  R^2   + \frac{1}{\omega_0}\Big(\cos(\psi_\ast+\Psi)-\cos\psi_\ast+\Psi\sin\psi_\ast\Big) + \omega_0^2 \varkappa^{3/5} \eta^{-3/5} \frac{R^3}{3}-\eta^{-1}\frac{R\Psi}{5} \\
                        &   & - \frac{m\rho_\ast}{2\omega_0} \Big( \cos(2\psi_\ast+2\Psi+\nu)-\cos(2\psi_\ast+\nu)-2\Psi \sin(2\psi_\ast+\nu)\Big)\\
                        &   & -\varkappa^{1/5} m \eta^{-1/5} \Big(\cos(2\psi_\ast+2\Psi+\nu)-\cos(2\psi_\ast+\nu)\Big)\frac{R}{2}
\end{eqnarray*}
and
\begin{eqnarray*}
    F  & = & \eta^{-1}\frac{\Psi}{5} \\
         &   & + \varkappa^{1/5}\eta^{-1/5} \Big[\frac{\cos(\psi_\ast+\Psi)}{\rho_\ast+\omega_0\varkappa^{1/5}\eta^{-1/5}R}-\frac{\cos\psi_\ast}{\rho_\ast} + \frac{m}{2} \big(\cos(2\psi_\ast+\nu)-\cos(2\psi_\ast+2\Psi+\nu)\big) \Big].
\end{eqnarray*}
The construction of the Lyapunov function proposed here is based on the asymptotic behavior of the right-hand side of system \eqref{ham} as $\eta\to\infty$. By taking into account \eqref{asser} one can readily write out the asymptotics of the functions $H(R,\Psi,\eta)$ and $F(R,\Psi,\eta)$:
\begin{gather*}
    H(R,\Psi,\eta)=H_0(R,\Psi)+ H_1(\Psi)\eta^{-2/5}+\mathcal O(\eta^{-3/5}), \\
    F(R,\Psi,\eta)=F_1 (\Psi)\eta^{-3/5}+F_2(\Psi)\eta^{-1}+\mathcal O(\eta^{-6/5}),
\end{gather*}
where
\begin{gather*}
    H_0 = \omega_0 \lambda^{1/2} R^2 + \frac{1}{\omega_0}\int\limits_0^\Psi \mathcal P(\psi_0+\phi;\delta,\nu)\,d\phi,\quad
    H_1 = \frac{\psi_1\varkappa^{2/5}}{\omega_0} \Big(\mathcal P(\psi_0+\Psi;\delta,\nu)- \mathcal P'(\psi_0;\delta,\nu) \Psi \Big), \\
    F_1 = \frac{\varkappa^{3/5}}{\lambda^{1/2}} \int\limits_0^\Psi \mathcal P(\psi_0+\phi;\delta,\nu)\,d\phi,\quad
    F_2 = \frac{\Psi}{5} + \frac{\psi_1\varkappa }{\lambda^{1/2}} \mathcal P(\psi_0+\Psi;\delta,\nu).
\end{gather*}
Note that all asymptotic estimates written out here and bellow in the form $\mathcal O(\eta^{-q})$ and $\mathcal O(d^q)$, $(d=\sqrt{R^2+\Psi^2})$, are uniform with respect to $(R,\Psi,\eta)$ in the domain $\mathcal D (d_\ast,\eta_\ast)=\{ (R,\Psi,\eta)\in\mathbb R^3: d<d_\ast, \eta>\eta_\ast \}$, where $d_\ast,\eta_\ast,q={\hbox{\rm const}}>0$.

A Lyapunov function candidate is constructed of the
form:
\begin{gather}
    \label{LF}
    V(R,\Psi,\eta)=\frac{1}{\omega_0 \lambda^{1/2}} \Big[ H(R,\Psi,\eta)+v_1(R,\Psi)\eta^{-3/5}+v_2(R,\Psi)\eta^{-1}\Big],
\end{gather}
where
\begin{gather*}
    v_1=\varkappa^{3/5} R\Big[\frac{2\omega_0^2 R^2}{3} + \lambda^{-1/2}\int\limits_0^\Psi \mathcal P(\psi_0+\phi;\delta,\nu)\, d\phi\Big], \quad v_2= -\frac{R\Psi}{10}.
\end{gather*}
Since $H_0(R,\Psi)=\omega_0 \lambda^{1/2}(R^2+\Psi^2)+\mathcal O(d^3)$ and $v_i(R,\Psi)=\mathcal O(d^2)$ as $d\to 0$, then for all  $0<\sigma<1$ there exist $d_0>0$ and $\eta_0>0$ such that
\begin{gather*}
    (1-\sigma)d^2\leq V(R,\Psi,\eta)\leq (1+\sigma)d^2
\end{gather*}
for $(R,\Psi,\eta)\in\mathcal D(d_0,\eta_0)$. The derivatives of  $H(R,\Psi,\eta)$, $v_1(R,\Psi)$, and $v_2(R,\Psi)$ with respect to $\eta$ along the trajectories of system \eqref{ham} have the following asymptotics:
\begin{eqnarray*}
    \frac{dH}{d\eta}\Big|_\eqref{ham} & = & \eta^{-3/5} \frac{\varkappa^{3/5}}{\omega_0^2 \lambda} \mathcal P(\psi_0+\Psi;\delta,\nu)\int\limits_0^\Psi \mathcal P(\psi_0+\phi;\delta,\nu)\,d\phi  - \eta^{-1} \frac{2\Psi^2}{5}+\\
                        & & + \mathcal O(d^3) \mathcal O(\eta^{-1})+\mathcal O(d^2)\mathcal O(\eta^{-6/5}), \\
    \frac{dv_1}{d\eta}\Big|_\eqref{ham} & = & -\frac{\varkappa^{3/5}}{\omega_0^2 \lambda} \mathcal P(\psi_0+\Psi;\delta,\nu)\int\limits_0^\Psi \mathcal P(\psi_0+\phi;\delta,\nu)\,d\phi + \mathcal O(d^3)\mathcal O(\eta^{-2/5})+\mathcal O(d^2)\mathcal O(\eta^{-3/5}), \\
    \frac{dv_2}{d\eta}\Big|_\eqref{ham} & =& \frac{\Psi^2-R^2}{5}+\mathcal O(d^3)+\mathcal O(d^2)\mathcal O(\eta^{-2/5}).
\end{eqnarray*}
These formulas are used to calculate the expression for the total derivative of the function $V(R,\Psi,\eta)$, which happens to have a sign-definite leading term of the asymptotics:
\begin{gather*}
\frac{dV}{d\eta}\Big|_\eqref{ham} = -  \frac{1}{5\eta} [R^2+\Psi^2+\mathcal O(d^3)]+\mathcal O(d^2)\mathcal O(\eta^{-6/5}).
\end{gather*}
Since the remainders in the latter expression can be made arbitrarily small, it follows that for all $0<\sigma<1$ there exist $d_0>0$ and $\eta_0>0$ such that
\begin{gather*}
    \frac{dV}{d\eta}\Big|_\eqref{ham} \leq - \frac{(1-\sigma)d^2}{5\eta}
\end{gather*}
for all $(R,\Psi,\eta)\in\mathcal D(d_0,\eta_0)$.
In addition, the Lyapunov function has the following property: for all $0<\epsilon<d_0$ there exist $\delta(\epsilon)=\epsilon\sqrt{(1-\sigma)/(2+\sigma)}$ such that
\begin{gather*}
\sup_{d<\delta(\epsilon),\eta>\eta_0} V(R,\Psi,\eta)\leq (1+\sigma) \delta^2(\epsilon)< (1-\sigma) \epsilon^2 \leq \inf_{d=\epsilon,\eta>\eta_0}V(R,\Psi,\eta).
\end{gather*}
The last estimates and the negativity of the total derivative of the function $V(R,\Psi,\eta)$ ensure that any solution of system \eqref{ham} with initial data $\sqrt{R^2(\eta_0)+\Psi^2(\eta_0)}\leq \delta(\epsilon)$ cannot leave $\epsilon$-neighborhood of the equilibrium $(0,0)$ as  $\eta>\eta_0$. Therefore, the fixed point $(0,0)$ is stable as $\eta>\eta_0$. The stability on the finite time interval $(0,\eta_0]$ follows from the theorem on the continuity of the solution to the Cauchy problem with respect to the initial data.

Let us show that the fixed point $(0,0)$ is asymptotically stable. Indeed, consider the solution $R(\eta)$, $\Psi(\eta)$ to system \eqref{ham} with initial data $\sqrt{R^2(\eta_0)+\Psi^2(\eta_0)}\leq d_0$, then the function $v(\eta)\equiv V(R(\eta),\Psi(\eta),\eta)$ satisfies the inequality:
\begin{gather}
    \label{LFD}
        \frac{dv}{d\eta}\leq -\frac{ \varsigma v}{\eta}, \quad \varsigma=\frac{(1-\sigma)}{5(1+\sigma)}>0.
\end{gather}
Integrating the last expression with respect to $\eta$, we obtain $0\leq v(\eta)\leq v_0 \eta^{-\varsigma}$ with the parameter $v_0$, depending on $d_0$ and $\eta_0$. This implies the asymptotic estimate: $R^2(\eta)+\Psi^2(\eta)=\mathcal O(\eta^{-\varsigma})$ as $\eta>\eta_0$. Returning to the original variables we derive asymptotic estimates \eqref{asth} with $\tau_0=\eta_0^{4/5}\varkappa^{-4/5}$ for solutions to system \eqref{MS} with initial data $(r_0,\phi_0)$ from a neighbourhood of  the point $(\rho_\ast(\tau_0),\psi_\ast(\tau_0))$.
 \end{proof}

Let us consider again the case $\nu=0$. If $|\delta|<1/2$, system \eqref{MS} has two particular autoresonant solutions with $\psi_0=0$ and $\psi_0=\pi$, and if $|\delta|>1/2$, there is an additional pair of solutions with $\psi_0=\pm\arccos(2\delta)^{-1}$. It is easy to calculate that $\mathcal P'(0;\delta,0)=2\delta-1$. Hence the unstable solution $\rho_\ast(\tau)$, $\psi_\ast(\tau)$ with $\psi_0=0$ becomes asymptotically stable when the parameter $\delta=\mu \lambda^{1/2}$ exceeds the threshold value $\delta_0=1/2$. At the same time, $\mathcal P'(\pi;\delta,0)=2\delta+1$, then the stability of the solution with $\psi_0=\pi$ persists as $\delta>-1/2$ and is lost when $\delta<-1/2$. However, since $\mathcal P'(\pm \arccos(2\delta)^{-1};\delta,0)=(1-4\delta^2)/(2\delta)$, the additional solutions are unstable as $\delta>1/2$. Thus the points $(\delta_0,0)\in\Gamma_{+}$ and $(-\delta_0,0)\in\Gamma_{-}$ correspond to non-autonomous version of the center-saddle bifurcation: stable and unstable solutions coalesce and disappear. The situation is similar when $\nu\neq 0$, and the threshold value $\delta_\nu$ for unstable solutions can be found from the equation $\mathcal P'(\psi_0;\delta_\nu,\nu)=0$.

\section{Asymptotic analysis}
\label{sec4}

The stability of the particular solutions $\rho_\ast(\tau)$, $\psi_\ast(\tau)$ ensures the existence of a two-parameter family of autoresonance solutions. At this section we construct the asymptotics for
such solutions by averaging method~\cite{AKN06} with some modifications, related to the involvement of the Lyapunov function, constructed in the previous section.

\begin{Th}
If $\mathcal P'(\psi_0;\delta,\nu)>0$, then system \eqref{MS} has two-parameter family of autoresonant solutions $\rho(\tau;a,\varphi)$, $\psi(\tau;a,\varphi)$ with the asymptotics
\begin{gather}
\label{genas}
\rho(\tau)=\sqrt{\lambda \tau}+\tau^{-1/4}\sum_{k=1}^\infty \rho_k(S;a)\tau^{-k/8}, \quad \psi(\tau)=\psi_0+\sum_{k=1}^\infty \psi_k(S;a)\tau^{-k/8}, \quad \tau\to\infty,
\end{gather}
where the functions $\rho_k(S;a)$, $\psi_k(S;a)$ are $2\pi$-periodic in $S$, and the function $S(\tau)$ has the following asymptotics as $\tau\to\infty${\rm :}
\begin{gather*}
    S(\tau)=\varphi+\frac{8 \omega_0 }{5}\lambda^{1/2}\tau^{5/4}+\sum_{k=1}^4 c_{-k}(a)\tau^{k/4}+c_{-6}(a)\log \tau+\sum_{k=1}^\infty c_k(a)\tau^{-k/4}, \quad c_k(a)={\hbox{\rm const}},
\end{gather*}
$\rho_1=a\omega_0\cos S$, $\psi_1=a\sin S$, $c_{-4}=a^2  \mathcal P'''(\psi_0;\delta,\nu) /(16 \omega_0)$, $\omega_0=(\mathcal P'(\psi_0;\delta,\nu))^{1/2}(4\lambda)^{-1/4}$.
\end{Th}

\begin{proof}
We make change of variables \eqref{exch2} in system \eqref{MS} and study the asymptotics for solutions to near-Hamiltonian system \eqref{ham} in a neighbourhood of the stable fixed point $(0,0)$.
Let us consider the Hamiltonian system:
\begin{gather*}
    \frac{dR}{d\eta}=-\partial_\Psi H_0(R,\Psi), \quad \frac{d\Psi}{d\eta}=\partial_R H_0(R,\Psi),
\end{gather*}
where $\displaystyle H_0(R,\Psi)=\lim_{\eta\to\infty} H(R,\Psi,\eta)$, $H_0(0,0)=0$. It follows from the properties of the function $H_0(R,\Psi)$ that the level lines $H_0(R,\Psi)=h$ define a family of closed curves on the phase space $(R,\Psi)$ parametrized by the parameter $h\in (0, h_0)$, $h_0={\hbox{\rm const}}>0$ (see Fig.~\ref{LevL}). To each closed curve there corresponds a periodic solution $R_0(\eta,h)$, $\Psi_0(\eta,h)$ of period $T(h)=2\pi/\omega(h)$,  where $\omega(h)=2\omega_0\lambda^{1/2}+h\mathcal P'''(\psi_0;\delta,\nu)/(16 \omega_0^2\lambda^{1/2})+\mathcal O(h^2)$ as $h\to 0$.
\begin{figure}
\vspace{-2ex} \centering \subfigure[]{
\includegraphics[width=0.45\linewidth]{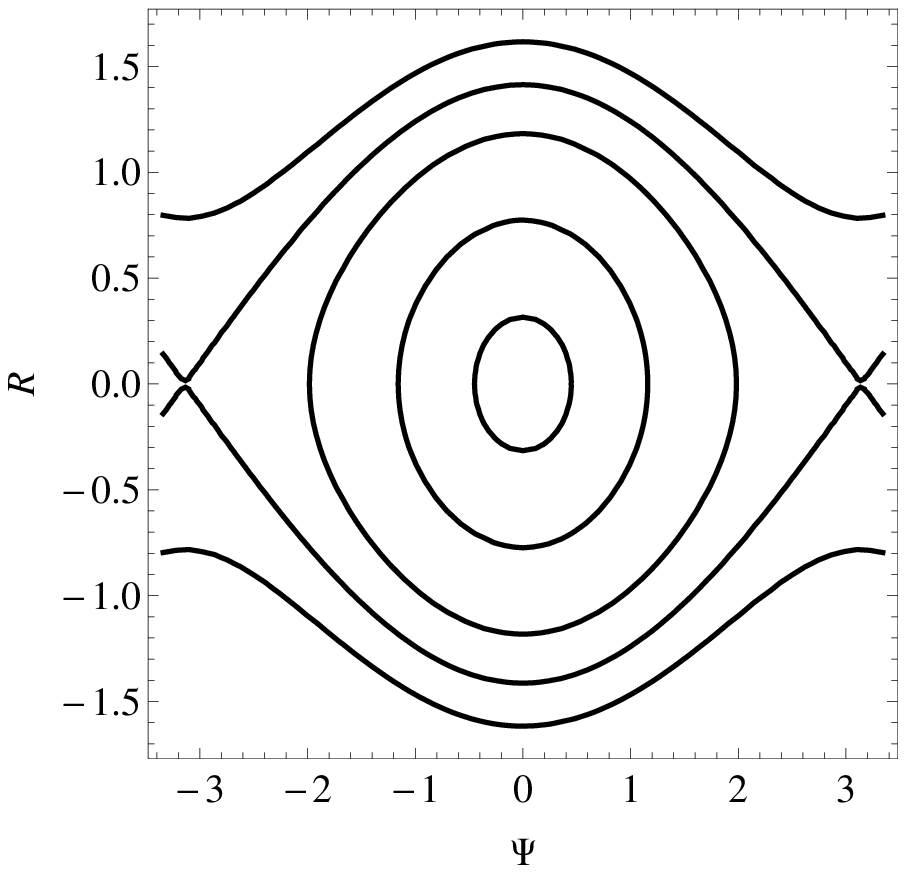}}
\hspace{4ex}
\subfigure[]{
\includegraphics[width=0.45\linewidth]{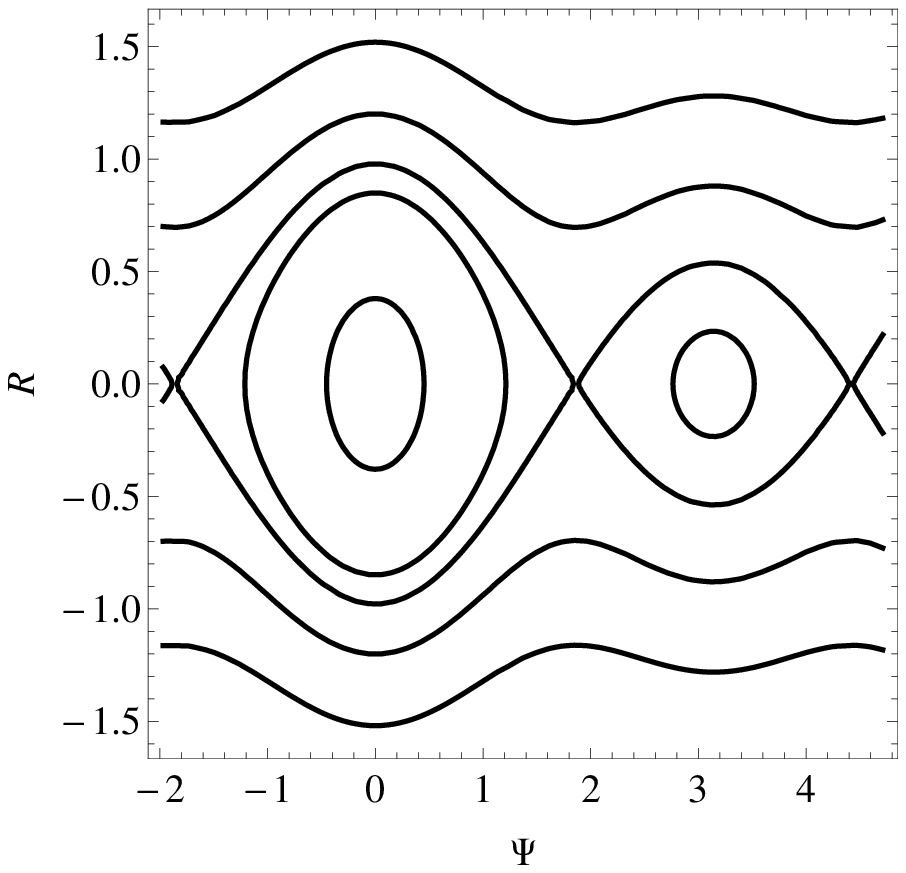}
}
\caption{Level lines of the Hamiltonian $H_0(R,\Psi)$; $\lambda=1$, $\nu=0$, $\psi_0=\pi$, (a) $\delta=0$, (b) $\delta=1$.} \label{LevL}
\end{figure}

Define auxiliary $2\pi$-periodic functions $r(s,h)=R_0(s/\omega(h),h)$ and $p(s,h)=\Psi_0(s/\omega(h),h)$, satisfying the equations:
\begin{gather*}
    \omega(h)\frac{\partial r}{\partial s}=-\partial_p H_0(r,p), \quad \omega(h)\frac{\partial p}{\partial s}=\partial_r H_0(r,p),
\end{gather*}
and consider the transformation of variables
\begin{gather}
\label{exch1}
    R(\eta)=r(s(\eta),h(\eta)), \quad \Psi(\eta)=p(s(\eta),h(\eta))
\end{gather}
in system \eqref{ham}. The Jacobian of this transformation is different from zero:
\begin{gather*}
    \begin{vmatrix}
        \partial_h r& \partial_h p \\
        \partial_s r & \partial_s p
    \end{vmatrix} = \frac{1}{\omega(h)}\neq 0,
\end{gather*}
It can easily be checked that in new variables system \eqref{ham} takes the form:
\begin{gather}
\label{sys2}
    \frac{dh}{d\eta}=-\omega(h) \Big[\partial_s \widetilde H(s,h,t) +\partial_s r \widetilde F(s,h,t)\Big], \quad
    \frac{ds}{dt}=\omega(h)\Big[\partial_h \widetilde H(s,h,t) + \partial_h r \widetilde F(s,h,t)\Big],
\end{gather}
where $\widetilde F(s,h,\eta)\equiv F(r(s,h),p(s,h),\eta)$ and $\widetilde H(s,h,\eta)\equiv H(r(s,h),p(s,h),\eta)$ are $2\pi$-periodic functions with respect to $s$. Moreover, system \eqref{sys2} is asymptotically Hamiltonian with decaying non-Hamiltonian terms as $\eta\to\infty$. Solutions to such systems can be investigated by the averaging method. To simplify asymptotic constructions, we make one more transformation of variables. Introduce a new dependent variable $v(\eta)$ associated with the Lyapunov function \eqref{LF} such that
\begin{gather}
\label{exchLF}
v(\eta)=\widetilde V(s(\eta),h(\eta),\eta),
\end{gather}
where $\widetilde V(s,h,\eta)\equiv V(r(s,h),p(s,h),\eta)$ is $2\pi$-periodic function in $s$. It follows from the properties of the Lyapunov function that $\widetilde V$ and its total derivative with respect to $\eta$ have the following asymptotics
\begin{gather*}
    \widetilde V(s,h,\eta)= \frac{h}{\omega_0\lambda^{1/2}}+\mathcal O(h)\mathcal O(\eta^{-3/5}), \quad
    \frac{d\widetilde V}{d\eta}\Big|_{\eqref{ham}}=- \frac{1}{5\eta}\Big [\frac{h}{\omega_0\lambda^{1/2}}+\mathcal O(h^{3/2})\Big]+\mathcal O(h)\mathcal O(\eta^{-6/5})
\end{gather*}
 as $\eta>\eta_0$, $h<h_0$, uniformly with respect to $s\in\mathbb R$. Hence the transformation $(s,h)\mapsto (s,v)$ is invertible: $\partial_h \widetilde V \neq 0$. System \eqref{sys2} in new variables takes the form:
\begin{gather}
\label{sys1}
    \frac{dv}{d\eta}= W(s,v,\eta), \quad \frac{ds}{d\eta}=U(s,v,\eta),
\end{gather}
where
\begin{gather}
\begin{split}
\label{WU}
    W(s,\widetilde V(s,h,\eta),\eta)&\equiv \big[-\partial_r V\partial_p H+\partial_p V (\partial_r H+F) +\partial_\eta V\big](r(s,h),p(s,h),\eta), \\
     U(s,\widetilde V(s,h,\eta),\eta)&\equiv \omega(h)\big[\partial_h \widetilde H+\partial_h r \widetilde F\big](s,h,\eta).
    \end{split}
\end{gather}
It is not difficult to deduce from \eqref{WU} and \eqref{LFD} the asymptotics of the functions $W(s,v,\eta)$ and $U(s,v,\eta)$ as $\eta\to\infty$:
\begin{gather*}
      W(s,v,\eta)=\eta^{-1}\sum_{k=0}^\infty  w_k(s,v) \eta^{-k/5}, \quad     U(s,v,\eta)=u_0(v)+\sum_{k=1}^\infty u_k(s,v) \eta^{-k/5},
\end{gather*}
where $w_k(s,v)$ and $u_k(s,v)$ are $2\pi$-parametric functions with respect to $s$, $w_0(s,v)=-v/5 +\mathcal O(v^{3/2})$, $u_0(v)=\omega(0)+\omega_0\lambda^{1/2}\omega'(0)v+\mathcal O(v^2)$ as $v\to 0$.
To construct asymptotic solutions to system \eqref{sys1}, it is convenient to single out a Hamiltonian part: \begin{gather}
\label{WUHam}
    \frac{dv}{d\eta}= -\partial_s Q(s,v,\eta) + G(s,v,\eta), \quad \frac{ds}{d\eta}=\partial_v Q(s,v,\eta),
\end{gather}
where
\begin{gather*}
    Q(s,v,\eta)=\int  U(s,v,\eta)\, dv, \quad  G(s,v,\eta)=W(s,v,\eta)+\partial_s Q(s,v,\eta).
\end{gather*}
The asymptotic solution to the first equation in \eqref{WUHam} is sought in the form~\cite{LK08JMS}:
$v(\eta)=\hat v(\eta)+\vartheta (s,\hat v(\eta),\eta)$, where
$\hat v(\eta)$ is determined from the averaged equation
\begin{gather}
\label{averA}
    \frac{d\hat v}{d\eta}= \langle G(s,\hat v+\vartheta(s,\hat v,\eta),\eta)\rangle_s=\langle W(s,\hat v+\vartheta(s,\hat v,\eta),\eta)\rangle_s,
\end{gather}
and $\vartheta (s,\hat v,\eta)$ is $2\pi$-periodic function with respect to $s$ such that
\begin{gather*}
    \langle \vartheta (s,\hat v,\eta)\rangle_s \stackrel{def}{=}\frac{1}{2\pi}\int\limits_0^{2\pi} \vartheta (s,\hat v,\eta) ds\equiv 0,
\end{gather*}
and satisfies the following equation
\begin{eqnarray*}
    \partial_\sigma Q(s,\hat v+\vartheta,\eta)+\partial_v Q(s,\hat v+\vartheta,\eta)\partial_s \vartheta+\partial_\eta \vartheta & = & G(s,\hat v+\vartheta,\eta) \\
        & & -[1+\partial_{\hat v} \vartheta] \langle G(s,\hat v+\vartheta(s,\hat v,\eta),\eta)\rangle_s.
\end{eqnarray*}
Note that the last equation can be integrated with respect to $s$ by choosing the constant of integration in such a way that the result has a zero average:
\begin{gather}
\label{int}
\begin{split}
    Q(s,\hat v+\vartheta(s,\hat v,\eta),\eta)-\langle Q(s,\hat v+\vartheta(s,\hat v,\eta),\eta)\rangle_s+\partial_\eta \int\vartheta(s,\hat v,\eta)\,ds  =  \\
   = \int \Big[W(s,\hat v+\vartheta,\eta)-  \langle W(s,\hat v+\vartheta(s,\hat v,\eta),\eta)\rangle_s \Big]ds \\
    - \langle W(s,\hat v+\vartheta(s,\hat v,\eta),\eta)\rangle_s \partial_{\hat v} \int \vartheta(s,\hat v,\eta) ds.
\end{split}
\end{gather}
The asymptotic for $\vartheta(s,\hat v,\eta)$ is sought in the form:
\begin{gather*}
    \vartheta (s,\hat v,\eta)=\eta^{-1}\sum_{k=1}^\infty \vartheta_k(s,\hat v)\eta^{-k/5} .
\end{gather*}
Substituting this series into equation \eqref{int}, and equating the terms of the same power of $\eta$, we obtain the following chain of equations: $u_0(\hat v)   \vartheta_k(s,\hat v) = \Lambda_k (s,\hat v)$, $k\geq 1$, where each function $\Lambda_k(s,\hat v)$ is expressed through $\vartheta_1$, $\dots$, $\vartheta_{k-1}$ such that $\langle \Lambda_k(s,\hat v)\rangle_s =0$. For example,
\begin{eqnarray*}
\Lambda_1(s,\hat v) & = &
\big\langle q_1(s,\hat v)\big\rangle_s-q_1(s,\hat v),\\
\Lambda_{2}(s,\hat v) & = & \big \langle q_2(s,\hat v)\big\rangle_s-q_2(s,\hat v)  + \big \langle u_1(s,\hat v)\vartheta_1(s,\hat v)\big\rangle_s-u_1(s,\hat v)\vartheta_1(s,\hat v)  \\
& & - \frac{u_0'(\hat v)}{2}\Big(\vartheta^2(s,\hat v)-\big\langle \vartheta^2(s,\hat v)\big\rangle_s\Big),
\end{eqnarray*}
where $q_k(s,v)=\int u_k(s,v)\,dv$. Thus, all coefficients $\vartheta_k$ are uniquely determined in the class of $2\pi$-periodic functions with zero average $\langle \vartheta_k(s,\hat v)\rangle_s=0$.

The next step consists in the asymptotic integration of the equation for the angle $s$. The solution is sought in the form $s(\eta)= \hat s(\eta)+\theta(\hat s(\eta),\hat v(\eta),\eta)$, where
\begin{gather}
\label{avs}
    \frac{d \hat s}{d\eta}=\big\langle \partial_v Q\big(\hat s+\theta(\hat s,\hat v,\eta),\hat v+\vartheta(\hat s+\theta(\hat s,\hat v,\eta),\hat v,\eta),\eta\big)\big\rangle_{\hat s},
\end{gather}
and the function $\theta(\hat s,\hat v,\eta)$ satisfies the following equation:
\begin{gather}
\label{psieq}
\begin{split}
   \big(\partial_{\hat s}\theta +1 \big) \big\langle \partial_v Q\big(\hat s+\theta(\hat s,\hat v,\eta),\hat v+\vartheta(\hat s+\theta(\hat s,\hat v,\eta),\hat v,\eta),\eta\big)\big\rangle_{\hat s} +\partial_\eta\theta = \\
   = \partial_v Q\big(\hat s+\theta(\hat s,\hat v,\eta),\hat v+\vartheta(\hat s+\theta(\hat s,\hat v,\eta),\hat v,\eta),\eta\big) - \partial_{\hat v} \theta \langle W(s,\hat v+\vartheta(s,\hat v,\eta),\eta)\rangle_s,
   \end{split}
\end{gather}
with the additional condition: $\langle \theta(\hat s,\hat v,\eta)\rangle_{\hat s}=0$. Asymptotics for $\theta(\hat s,\hat v,\eta)$ is constructed in the form:
\begin{gather*}
    \theta(\hat s,\hat v,\eta)=\sum_{k=1}^{\infty} \theta_k(\hat s,\hat v)\eta^{-k/5}.
\end{gather*}
Substituting the series into equation \eqref{psieq} and grouping the expressions of the same power of $\eta$ give the following chain of differential equations: $u_0(\hat v)\partial_{\hat s} \theta_k(\hat s,\hat v)=\Delta_k(\hat s,\hat v) - \langle \Delta_k(\hat s,\hat v)\rangle_{\hat s}$, $k\geq 1$. Note that each function $\Delta_k(\hat s,\hat v)$ is expressed through $\theta_1$, $\dots$, $\theta_{k-1}$. For example, \begin{eqnarray*}
    \Delta_1(\hat s,\hat v)&=&u_1(\hat s,\hat v)+u_0'(\hat v) \vartheta_1(\hat s,\hat v), \\
    \Delta_2(\hat s,\hat v)&=&
        u_2(\hat s,\hat v)+u_0'(\hat v)\big( \vartheta_2(\hat s,\hat v)+\partial_{ s}\vartheta_1(\hat s,\hat v)\theta_1(\hat s,\hat v)\big)+ \frac{u_0''(\hat v)}{2}\vartheta^2_1(\hat s,\hat v)\\
    & &  + \partial_v u_1(\hat s,\hat v)\vartheta_1(\hat s,\hat v)+\partial_{ s} u_1(\hat s,\hat v)\theta_1(\hat s,\hat v) \\ & & -\partial_{ s}\theta_1(\hat s,\hat v) \big\langle u_0'(\hat v)\vartheta_1(\hat s,\hat v)+u_1(\hat s,\hat v) \big\rangle_{\hat s}.
\end{eqnarray*}
It follows that all coefficients $\theta_k$ are uniquely determined in the class of $2\pi$-periodic functions with $\langle \theta_k(\hat s,\hat v)\rangle_{\hat s}=0$.

In the last step, we integrate the averaged equations \eqref{averA} and \eqref{avs}. Let us remark that the stability of the trivial solution to system \eqref{ham} ensures that $v(\eta)\to 0$ as $\eta>\eta_0$. Therefore, we can use the asymptotic behavior of the function $W(s,v,\eta)$ as $v\to 0$ and $\eta\to\infty$ in the asymptotic integration of the averaged equations. Thus we have
\begin{gather*}
    \frac{d \hat v}{d\eta}=- \frac{1}{5\eta}[\hat v+\mathcal O(\hat v^{3/2})][1+\mathcal O(\eta^{-1/5})].
\end{gather*}
It follows from the last equation that $\hat v(\eta)=v_0\eta^{-1/5}[1+o(1)]$ with an arbitrary parameter $v_0>0$. Moreover, there is a complete asymptotic expansion:
\begin{gather*}
   \hat v(\eta)=v_0\eta^{-1/5}+\sum_{k=1}^\infty v_k(v_0) \eta^{-(k+1)/5}, \quad \eta\to\infty, \quad v_k(v_0)={\hbox{\rm const}}.
\end{gather*}

The averaged equation \eqref{avs} for the angle is integrated trivially:
\begin{gather*}
    \hat s(\eta)=s_0+\int\limits  \big\langle \partial_v Q\big(\hat s+\theta(\hat s,\hat v(\eta),\eta),\hat v(\eta)+\vartheta(\hat s+\theta(\hat s,\hat v(\eta),\eta),\hat v(\eta),\eta),\eta\big)\big\rangle_{\hat s}\, d\eta,
\end{gather*}
where $s_0$ is the integration parameter. The asymptotics of $\hat s(\eta)$ is defined by the expansions for $\hat v$, $\vartheta$, and $\theta$. Thus we have
\begin{gather*}
    \hat s(\eta)=s_0+\sum_{k=1}^{5} s_{-k}(v_0)\eta^{k/5}  +s_{-6}(v_0)\log \eta +  \sum_{k=1}^\infty s_k(v_0)\eta^{-k/5}, \quad \eta\to\infty,
\end{gather*}
where $s_k(v_0)={\hbox{\rm const}}$. In particular, $s_{-5}=2\omega_0\lambda^{1/2}$ and $s_{-4}=5v_0\omega_0 \lambda^{1/2} \omega'(0)/4$. The justification of the constructed asymptotics follows from~\cite{LK08JMS}.

The transformation formulas \eqref{exch1} and \eqref{exchLF} allow us to reconstruct the asymptotics for solutions to system \eqref{ham}:
\begin{gather*}
    R(\eta)=  v_0^{1/2} \eta^{-1/10}\big[\cos \hat s(\eta) +\mathcal O(\eta^{-1/5})\big], \quad
    \Psi(\eta)=v_0^{1/2} \eta^{-1/10} \big[\sin \hat s(\eta) +\mathcal O(\eta^{-1/5})\big].
\end{gather*}
Returning to the original variables we obtain the result of the theorem.
\end{proof}

Note that using the Lyapunov function in the averaging method allows us to immediately write down the averaged equation for the action variable, that determines the structure of a complete asymptotic expansion. The proposed method for constructing two-parameter family asymptotic solutions can be applied to the study of various nonlinear non-autonomous systems that can be reduced to a near-Hamiltonian form. The construction of the Lyapunov functions for such systems was discussed in~\cite{LKOS13}.

\section{Conclusion}

The results obtained for the model system can be used to describe the capture of autoresonance in systems of the form \eqref{ex} with combined excitation. It can easily be checked that the bifurcation parameter $\delta=\mu\lambda^{1/2}$ is expressed in terms of system \eqref{ex} as follows $\delta=2\beta\varepsilon^{-1} \alpha^{1/2} (3\gamma)^{-1/2} $. Therefore, if $\delta\neq \delta_\nu$, there exist a stable autoresonant mode such that
\begin{gather*}
u(t)= \kappa \rho(\varepsilon t) \cos \Big(\psi(\varepsilon t)-\phi(t)\Big)+\mathcal O(\varepsilon), \quad 0\leq t\leq \mathcal O(\varepsilon^{-1}),
\end{gather*}
where $\rho(\tau)=\sqrt{\lambda\tau}+\mathcal O(\tau^{-3/8})$ and $\psi(\tau)=\psi_0+\mathcal O(\tau^{-1/4})$ as $\tau\to\infty$ with $\psi_0$ satisfying $\mathcal P(\psi_0;\delta,\nu)=0$ and $\mathcal P'(\psi_0;\delta,\nu)>0$. At the same time, the autoresonant mode with $\psi_0$ such that $\mathcal P'(\psi_0;\delta,\nu)<0$ is not observed because it is unstable. However, if $\delta$ passes through the value $\delta_\nu$, this mode becomes stable.

In summary, we have investigated the model system of equations describing the capture into autoresonant for nonlinear oscillators with combined parametric and external excitation. We have considered the isolated autoresonant solutions with a special asymptotics at infinity and have investigated the dependence of their stability on the values of the perturbation parameters. In particular, it was shown that the unstable autoresonant solutions become stable as the bifurcation parameter passes through a threshold value. The behaviour of solutions at the bifurcation points has not been considered here. This will be discussed elsewhere.

\section*{Acknowledgements}
This research was supported by the DAAD. The author is grateful to the University of Potsdam for hospitality and Professor Nikolai Tarkhanov for discussions.

\end{document}